\documentclass{llncs}
\usepackage{amsmath,amssymb,amsfonts,stmaryrd}
\usepackage{graphicx}
\usepackage{xcolor}
\usepackage{hyperref}
\usepackage[utf8]{inputenc}
\usepackage[T1]{fontenc}
\usepackage{booktabs}
\usepackage{xspace}

\usepackage{algorithm}
\usepackage{algorithmicx}
\usepackage[noend]{algpseudocode}
\usepackage{siunitx}

\newif\ifcomments%
\commentstrue%

\newcommand{\R}{\ensuremath{\mathbb{R}}}
\newcommand{\N}{\ensuremath{\mathbb{N}}}

\newcommand{\Pc}{{\ensuremath{\cal P}}}
\newcommand{\Rp}{\mathbb{R}_{+}}

\newcommand{\bbm}{\begin{bmatrix}}
\newcommand{\ebm}{\end{bmatrix}}
\newcommand{\PIVP}[3]{\frac{d}{dt}\bbm#1\ebm=\bbm#2\ebm,\quad\bbm#3\ebm_{t=0}}
\newcommand{\PODE}[2]{\frac{d}{dt}\bbm#1\ebm=\bbm#2\ebm}

\newcommand{\x}{\ensuremath{\vec x}}

\newtheorem{thm}{Theorem}

\newtheorem{lmm}{Lemma}

\newcommand{\comment}[1]{}

\begin{document}

\title{On the Complexity of Quadratization for Polynomial Differential Equations}
\author{Mathieu Hemery \and Fran\c{c}ois Fages \and Sylvain Soliman}
\institute{Inria Saclay Ile de France, EP Lifeware, Palaiseau, France}
\maketitle
\thispagestyle{plain}

\begin{abstract}
  Chemical reaction networks (CRNs) are a standard formalism used in chemistry and biology to reason about the dynamics of molecular interaction networks.
  In their interpretation by ordinary differential equations, CRNs provide a Turing-complete model of analog computation,
  in the sense that any computable function over the reals can be computed by a finite number of molecular species
  with a continuous CRN which approximates the result of that function
  in one of its components in arbitrary precision.
  The proof of that result is based on a previous result of Bournez et al.~on the Turing-completeness of polynomial ordinary differential equations
  with polynomial initial conditions (PIVP).
  It uses an encoding of real variables by two non-negative variables for concentrations,
   and a transformation to an equivalent quadratic PIVP (i.e.~with degrees at
   most 2) for restricting ourselves to at most bimolecular reactions.
  In this paper, we study the theoretical and practical complexities of the quadratic transformation. 
  We show that both problems of minimizing either the number of variables (i.e., molecular species)
or the number of monomials (i.e.~elementary reactions) in a quadratic transformation of a PIVP are NP-hard.
  We present an encoding of those problems in MAX-SAT and show the practical complexity of this algorithm
  on a benchmark of quadratization problems inspired from CRN design problems.
\end{abstract}

\section{Introduction}

Chemical reaction networks (CRNs) are a standard formalism used in chemistry and biology to reason about the dynamics of molecular interaction networks.  
A CRN over a vector $\x$ of molecular species is a finite set of formal chemical reactions
of the form $$r(\x)\xrightarrow{f(\x)} p(\x)$$
composed of a multiset $r(\x)$ of reactants (with multiplicity given by stoichiometric coefficients in $r$),
a multiset $p(\x)$ of products,
and a rate function $f(\x)$ on the quantities of reactants.
The structure of a CRN is the same as the structure of a Petri net, but the rate functions allow for the definition of
continuous-time dynamics in addition to their discrete dynamics:
in particular the stochastic semantics which interprets a CRN by a continuous-time Markov chain,
and the differential semantics which interprets a CRN by a system of ordinary differential equations (ODEs)~\cite{CSWB09ab,FS08tcs}.

In the differential semantics of a CRN $R=\{r_j(\x)\xrightarrow{f_j(\x)} p_j(\x)\}$,
one associates to each molecular species $x_i$ a concentration also noted $x_i$ by abuse of notation,
with the differential function
$$\frac{d x_i}{dt} = \sum_{j\in R} (p_j(x_i)-r_j(x_i)).f_j(\x).$$
Mass action law kinetics are monomial rate functions that lead to polynomial ODEs.
The other standard rate functions used such as Michaelis-Menten kinetics and Hill kinetics are traditionnally obtained by approximations of
mass action law systems~\cite{Segel84book}, and can thus be disregarded without loss of generality.

Collision theory shows however that the probabilities of reactions involving three or more reactants are negligible.
Hence from a mechanistic point of view, the restriction to reactions involving at most two reactant molecules is of practical importance.
We call an elementary CRN (ECRN) a CRN with at most bimolecular reactions and mass action law kinetics.
The restriction to at most bimolecular reactions leads to polynomial ODEs of degree at most 2.

With these restrictions, ECRNs have been shown to provide a Turing-complete model of analog computation,
in the sense that any computable function over the reals can be computed by an ECRN which approximates the result of that function
on one of its components in arbitrary precision~\cite{FLBP17cmsb}.
More precisely, we say that a CRN with a distinguished output species $x_1$ generates a function of time $f:\Rp\rightarrow\Rp$ from initial state $\x(0)$
if $\forall t\ x_1(t)=f(t)$.
A CRN with distinguished input and output species $x_0$ and $x_1$ computes a positive real function $f:\Rp\rightarrow\Rp$
from initial state $\x_i(0)=q(x_0(0))$ for some polynomial $q$ and $i\in\{1,\ldots,n\}$,
if for any initial concentration $x_0(0)$ of the input species,
the concentration of the output molecular species stabilizes at concentration $x_1=f(x_0(0))$.
The proof of Turing-completeness of ECRNs in~\cite{FLBP17cmsb} is based on a previous
result of Bournez et al.~in~\cite{BCGH06complexity} on the Turing-completeness of
polynomial ordinary differential equations with polynomial initial values (PIVPs) for
computing real functions~\cite{GC03}.  The proof for ECRNs uses on the one hand, on an encoding of real
variables $x$ by the difference of two non-negative variables $x^+$ and $x^-$ for
concentrations, and on the other hand, on a transformation of the PIVP to a quadratic PIVP~\cite{CPSW05ejde} computing the same
function but with degree at most 2.

In this paper, we study the quadratic transformation problem and its computational complexity.

\comment{
It is well-known that any solution of a PIVP is the solution of a PIVP of degree at most 2,
see for instance \cite{CPSW05ejde}.
However, to our knowledge, this problem 
has not been studied from a computational point of view up to now,
perhaps because the opposite is usually preferred, i.e.~using higher degrees to lower the dimension,
while here in the context of ECRN synthesis, higher degrees must be eliminated\footnote{
Another motivation for decreasing the degrees of polynomial ODEs
comes from the fact that higher degrees may introduce numerical instabilities in the
numerical integration of ODEs~\cite{CPSW05ejde}.
This has also been remarked in the context of CRNs,
for example when conservation laws (i.e.~Petri net place invariants) are used to decrease the dimension of the system
by replacing some variables by linear functions of other variables.}.
The following example motivates the study of the quadratic transformation problem as an optimization problem.
}

\begin{example}\label{Hill5}
The hill function of order $5$:
$$H_5(x) = \frac{x^5}{1+x^5}.$$
is an interesting example because it has been shown to provide a good approximation of the input/output function
of the MAPK signalling network which is an ubiquituous CRN structure present in all eukaryote cells and in several copies~\cite{HF96pnas}.
That function is a stiff sigmoid function which provides the MAPK network with a switch-like response to the input, ultrasensitivity and
an analog/digital converter function.
It is thus interesting to compare the MAPK network to the CRN design method above based on the mathematical definition of the $H_5$ function by ODEs.
Following~\cite{FLBP17cmsb}, one can easily check that the function $H_5(x)$ is computed by the following PIVP
noted in vectorial form for the differential equations\footnote{More precisely, the first two equations
  have for solution the Hill function of order 5 as a function of time $T$, and the last two equations has for effect to stop time $T$ at initial value $X(0)$}
and the initial conditions:
$$\PIVP{H \\ I \\ T \\ X}{5.I^2.T^4.X \\ -5.I^2.T^4.X \\ X \\ -X}{0 \\ 1 \\ 0 \\ x}.$$
For any positive value $X(0)=x$ in the initial condition, we have $\lim_{t \rightarrow \infty} H(t) = H_5(x)$.
However, this PIVP is of order $7$ and its direct implementation by CRN would involve non-elementary reactions with 7 reactants.
In this example, the proof of existence of a quadratic transformation for PIVPs given in~\cite{CPSW05ejde} introduces $29$ variables,
while the MAPK network involves 12 molecular species.
In this paper, we consider the quadratic transformation problem as an optimization problem which consists in minimizing
the dimension of the quadratic PIVP.
The optimization algorithm we propose
generates the following optimal ECRN for implementing $h_5(x)$ with only $7$ variables
   (named below by the monomial they represent, e.g.~\verb+it4x+ for $I.T^4.X$) and $11$ reactions with mass action law kinetics (MA):

 \begin{minipage}[t]{.5\textwidth}
        \centering
\begin{verbatim}
MA(5.0) for i+it4x=>h+it4x
MA(1.0) for x=>_
MA(1.0) for 2*x=>tx+2*x
MA(1.0) for tx=>_
MA(3.0) for 2*tx=>t3x+2*tx
MA(1.0) for t3x=>_
\end{verbatim}
\end{minipage}
 \begin{minipage}[t]{.5\textwidth}
\begin{verbatim}
MA(1.0) for ix=>_
MA(5.0) for it4x+ix=>it4x
MA(4.0) for ix+t3x=>it4x+ix+t3x
MA(1.0) for it4x=>_
MA(5.0) for 2*it4x=>it4x
\end{verbatim}
 \end{minipage}
\end{example}

In the following, we show that both problems of minimizing either the dimension
(i.e.~number of molecular species) or the number of monomials
(i.e.~number of reactions\footnote{While the correspondance between variables and species is exact,
the one between monomials and reactions is in fact more complicated if stoechiometric coefficients and
rate constants are exchanged when gathering the monomials appearing in the different
differential functions. In the following, we will nonetheless minimize monomials as a proxy for the
number of reactions.})
in a quadratic transformation of a PIVP are NP-hard.
The proof is by reduction of the vertex set covering problem (VSCP).
We present an algorithm based on an encoding in a MAX-SAT Boolean satisfiability problem,
and show its practicality on a benchmark of quadratization problems inspired from CRN design problems.

The rest of the paper is organized as follows.
In the next section, we define the quadratic transformation decision problem (QTDP)
as the problem of deciding whether there exists a PIVP quadratization of some given dimension $k$,
and the associated optimization problem (QTP) to determine the minimum number $k$ of variables.
We also consider the minimization of the number of monomials.
The difficulty of those problems are illustrated with some motivating examples.
We distinguish the succinct representation of the input PIVP by a list of monomials,
under which QTP is shown to be in NEXP,
from the non succinct representation by the full matrix of possible monomials of the input PIVP
under which QTP is shown to be in NP\@.
In Section~\ref{sec:maxsat},
we present an encoding of the QTP as a MAX-SAT Boolean satisfiability problem,
and derive from that encoding an algorithm to solve QTP and its variant for minimizing the number of monomials.

Then in Section~\ref{sec:np}, we show that the different QTP problems are NP-hard.
More precisely, we show that the decision problem in the non-succinct representation of the input PIVP is NP-complete by reduction of the Vertex Set Covering Problem (VSCP),
and we conjecture that the decision problem in the succinct representation is NEXP-complete
with the argument that some hard instances of QTP require an exponential number of variables in the size of the input PIVP.

Then in Section~\ref{sec:evaluation}, we study the practical complexity of QTP.
We propose a benchmark of PIVP quadratization problems inspired from CRN design problems,
and show the performance of the MAX-SAT algorithm on this benchmark\footnote{The benchmark and the implementation in BIOCHAM are available online in a Jupyter notebook at \url{https://lifeware.inria.fr/wiki/Main/Software\#CMSB20a}.}

\section{Quadratic Transformation of PIVPs}\label{QTP}

\subsection{Quadratic Projection Theorem}\label{Carothers}

A PIVP is a system of polynomial differential equations given with initial values.
Following the notations of~\cite{CPSW05ejde},
from $\mathcal{A}$ the set of real analytic functions, we say that $f\in\mathcal{A}$ is
\emph{projectively polynomial} if $f$ is a component of the solution of a PIVP. We note
$\Pc$ the set of such functions, and $\Pc_k(n)$ the subset of functions defined by a PIVP of dimension $n$ and degree at most $k$.
$\Pc_k$ will denote $\bigcup_{n \in \N}\Pc_k(n)$.

\begin{example}
The cosine function belongs to the class $\Pc_1(2)$ since it may be defined over $\R$ through the PIVP\@:
$$ \PIVP{x \\ y}{-y \\ x}{1 \\ 0}.$$
That notation will be kept throughout the article with the last element denoting the
initial condition of the PIVP (at $t=0$ by convention).
\end{example}

A folklore theorem of polynomial differential equation systems is that they can be
restricted to degree at most 2 without loss of generality on the generated functions:

\begin{thm}\label{quadratic}
	$\Pc = \Pc_2$:
	any function generated by a PIVP can be generated by a PIVP of degree at most two.
\end{thm}

The proof given in~\cite{CPSW05ejde} is based on Alg.~\ref{algo0} which consists in
introducing as many new variables as the number of possible monomials.

\begin{algorithm}

Input: PIVP with $n$ variables $\{x_1,\ldots,x_n\}$, and maximum power $d_j$ per variable.

Output: quadratic PIVP with same output function on variable $v_{1,0,\ldots,0}(t)$.

\begin{enumerate}

	\item Introduce the variables $v_{i_1,\ldots,i_n}=x_1^{i_1}x_2^{i_2},\ldots,x_n^{i_n}$ for
		all $i_j$, $0\le i_j\le d_j$, $1\le j\le n$ satisfying $i_k>0$ for some variable indice
		$k$;
	\item If the output variable $x_1$ has a maximum power $0$, add the variable
		$v_{1,0,\ldots}=x_1$. 
	\item Compute the derivatives of the $v$ variables as functions of the $x$ variables;
	\item Replace the monomials in the derivatives of the $v$ variables by monomials of the $v$ variables with degree at most $2$.
\end{enumerate}
\caption{Quadratization algorithm of Carothers et al.~\cite{CPSW05ejde}.}\label{algo0}
\end{algorithm}

While it is obvious that the
derivatives of the original variables can be rewritten by a sum of the new variables,
one must check that the derivatives of the new variables can be written in quadratic form.
Let $x_1, x_2 \ldots$ be the variables of the input PIVP, and
$d_n$ be the highest degree of $x_n$ among all the monomials of the input PIVP.
One new variable is introduced for each monomial $v_{\{i_1, \ldots, i_n\}} = \prod x_n^{i_n}$ that is
possible to construct with $i_n \in \{0, \ldots, d_n\}$ and at least one $i_n$ strictly
positive\footnote{One can remark that step $2$ in Alg.~\ref{algo0} was omitted in the
original proof of \cite{CPSW05ejde} but is necessary, as shown for instance for the Hill
function given in Sec.~\ref{sec:evaluation}.}.
It is then clear that the original function is still computed by the output PIVP of
Alg.~\ref{algo0} since we explicitely introduce it.  Furthermore, we can compute the
derivative of the new variables:
\begin{equation}
	\frac{d}{dt} \prod x_n^{i_n} = \sum_k \left( i_k \frac{d x_k}{dt} x_k^{i_k-1} \prod_{n
	\neq k} x_n^{i_n} \right),
\end{equation}
and it is enough to note that $v_{\{i_1, \ldots, i_{k-1}, \ldots, i_n\}}$ is one of the new
variables and that $\frac{d x_k}{dt}$ has only monomial of degree one in the new set of
variables by construction. This derivative is thus quadratic with respect to the new
variables.

\begin{proposition}\label{algo0exp}
  Alg.~\ref{algo0} introduces $O(d^{n})$ variables where $n$ is the number of variables and $d$ the maximum power in the original PIVP\@.
\end{proposition}

\begin{proof}
	For a PIVP of $n$ variables $x_i$ with highest degree $d_i$ and with a
	distinguished output variable $x_1$, Alg.~\ref{algo0} introduces $\prod_i
	(d_i+1) - 1 + \delta(d_1,0) = O(d^{n})$ variables where $\delta$ is the Kronecker delta which is
	$1$ iff $d_i=0$ and $0$ otherwise.
	The first term in the expression comes from the fact that each old variable $x_i$ may appear
	in the new set of variables with a power ranging from $0$ to $d_i$. The second term comes
	from the exclusion of the null variable, and the last one prevents us to delete the
	distinguished output variable if it does not appear in the derivatives.
\end{proof}

However, Alg.~\ref{algo0} may introduce much more variables than is actually needed
as already shown by Ex.~\ref{Hill5} and more precisely by the examples below.

\subsection{Examples}

\begin{example}
	Applying Alg.~\ref{algo0} to the PIVP $d_t x = x^k$ with the initial condition
	$x(t=0)=x_0$ would introduce $k$ variables for
	$x, x^2, \dots, x^{k}$.
But as it can be easily checked, that this PIVP can also be quadratized
with only two variables: $x, y = x^{k-1}$ with
$$ \PIVP{x \\ y}{xy \\ (k-1) y^2}{x_0 \\ x_0^{k-1}}.$$
\end{example}

In the example above, the number of variables needed does not depend on the degree of
the input PIVP. More generally it is not always the case that when the degree of a monomial
increases, the minimum number of variables in a quadratized form of the PIVP increases:
\begin{example}
The system\@:
$$ \PODE{x\\y}{y^3\\x^3 + x^2y^2} $$
needs $7$ variables ($x, y, xy, y^2, x^3, y^3, xy^2$).
When increasing the highest degree by one:
$$ \PODE{x\\y}{y^4\\x^4 + x^2y^2} $$
we need only $6$ variables ($x, y, x^3, y^3, x^2y, xy^2$).
But pursuing to increase:
$$ \PODE{x\\y}{y^5\\x^5 + x^2y^2} $$
needs now $9$ variables for example: $x, y, x^3, y^3, xy^2, x^4, y^4, x^3y, xy^3$.
This is still far less than the solution given by the mathematical proof with the $35$ variables 
of the monomials smaller than $x^5y^5$.
\end{example}

\begin{example}
On the system $ \PODE{x\\y}{y^3\\x^3} $,
our algorithm presented in the sequel returns the following solution with $5$ variables ${a=x, b=y, c=x^2, d=y^2, e=xy}$:
\begin{align}
	d_t a &= y^3 = bd, \\
	d_t b &= x^3 = ac, \\
	d_t c &= 2xy^3 = 2de, \\
	d_t d &= 2x^3y = 2ce, \\
	d_t e &= x^4 + y^4 = c^2 + d^2.
\end{align}
A critical aspect of the optimal solution is that it may contain
	monomials, like $xy$ here, that do not appear in the derivatives of the initial variables and could appear unnecessary at first glance.
\end{example}

\begin{example}\label{ex:exponential}
	Interestingly, the PODE 
	$$ \PODE{a\\b\\c}{b^2 + a^2b^2c^2\\c^2 + a^2b^2c^2\\a^2 + a^2b^2c^2}$$
	where each derivative is composed of the square of the next variable in addition to
	a long monomial formed
	with the square of all possible variables is among the ones needing the most
	variables. For this example, the optimal set found by our algorithm described in the sequel is:
	$$	\{a,b,c,a^2,b^2,c^2,abc,ab^2,ac^2,a^2b,a^2c,bc^2,b^2c,ab^2c^2,a^2bc^2,a^2b^2c\}, $$ 
	that is $16$ variables.

\end{example}

Although we have not been able to prove it, the previous example suggests that a quadratic
transformation may effectively need an exponential number of variables. We thus formulate
the following conjecture:

\begin{conjecture}\label{conjecture}
  The quadratization of PIVPs of the form:
	$$\frac{dx_i}{dt} = x_{i+1}^2 + \prod x_j^2, \quad x_i(t=0)=1,$$
	with $i\in(1,\ldots,n)$ and where $x_{n+1}$ denotes $x_1$, requires an exponential number of variables in $n$. 
  \end{conjecture}

\subsection{Quadratic Transformation Problems}

The quadratic transformation problem (QTP) is the optimization problem of determining the
minimum number of variables necessary to define an equivalent quadratic PIVP:

\begin{description}
\item {\bf instance:} A PIVP on $n$ variables $X=\{x_i\}_{0\le i\le n-1}$ with a distinguished output variable $x_0$.
\item {\bf output:} the minimum number $k$ of functions $f_j(X)$ 
  such that $\{x_0, f_j(X)\}$ defines an algebraically equivalent quadratic PIVP\@.
\end{description}
The associated decision problem (QTDP) is:
\begin{description}
\item {\bf instance:} A PIVP on variables $X=\{x_i\}$, a distinguished variable $x_0$ and an integer $k$
\item {\bf output:} existence of not of $k$ functions $f_j(X)$ 
  such that $\{x_0, f_j(X)\}_{1\le j\le k}$ defines an algebraically equivalent quadratic PIVP\@.
\end{description}

It is worth noting that the computational complexity of a decision problem may change drastically, for instance from NP to NEXP, according to the
succinct or not representation of the input~\cite{PY86ic}.  The representation of the input PIVP given above by
a list of symbolic functions is a succinct representation.  A non-succinct representation
of the input PIVP is given by the matrix of monomial coefficients $K: \R^n\times\R^m$ where $n$ is
the dimension of the PIVP and $m\le (d+1)^{n}$ the number of possible monomes to consider (Prop.~\ref{algo0exp}).

Let us denote by nsQTP and nsQTDP the non-succinct variants of the QTP and QTDP problems.

\begin{proposition}\label{inNP}
nsQTDP $\in$ NP. QTDP $\in$ NEXP\@.
\end{proposition}

\begin{proof} 
	By Prop.~\ref{algo0exp}, the size of a witness for a quadratic PIVP is less than the non-succinct representation of the input PIVP
	by the full matrix of possible monomials.
	Given such a witness quadratic PIVP
        one can check in polynomial time that it defines a quadratic PIVP algebraically equivalent to the original PIVP.
	For that, we just have to compute the derivatives of all the new variables expressed as functions
	of the old ones; then to express still with the old variables, all the monomials of
	degree 2 that may be formed with the new variables (an operation that is clearly
	quadratic in the number of variables); and finally to rewrite all the new derivatives with monomials or quadratics of the new variables.
	As each derivative contains only a linear number of monomials, we have a quadratic algorithm to check the validity of a
	witness, hence we have nsQTDP $\in$ NP\@.

        Now, in the succinct representation of the input PIVP by lists of monomials,
        the size of the witness is bound by an exponential in the size of the input PIVP, hence we simply get QTDP $\in$ NEXP\@.
\end{proof}

In the following (Thm.~\ref{thm:complete}), we show that nsQTDP is actually NP-complete, and thus nsQTP NP-hard,
and we conjecture that QTDP is NEXP-complete by extending our conjecture~\ref{conjecture} above to hard instances.

\section{MAX-SAT Encoding}\label{sec:maxsat}

The maximum satisfiability problem (MAX-SAT) is a generalization of the Boolean satisfiability
problem SAT, where some \emph{soft} clauses, that can be either true or false in a
solution, are added to a traditional (\emph{hard}) SAT problem, and where the
optimization problem of maximizing the number of soft clauses satisfied is
considered.

Alg.~\ref{algo0} can be reformulated in MAX-SAT form,
by expressing the constraints of QTP with Boolean clauses which lead to Alg.~\ref{algo:max-sat}.

\begin{algorithm}
\begin{enumerate}

  \item For each monomial \(m\) in the set \(M\) considered in Sec.~\ref{Carothers}
  (i.e., all those corresponding to variables $v$ of step 1 of Alg.~\ref{algo0}),
  introduce a Boolean variable \(x_{m}\) representing its presence in the
  reduced system.

  \item For each of those monomials, compute its derivative \(m'\)
  (same as step 2 of Alg.~\ref{algo0})

  \item For each monomial appearing in any \(m'\), compute all the ways to represent
  it as the product of 0 (constant case), 1 or 2 of the monomials of \(M\).

  \item Now add to the MAX-SAT model one hard clause imposing that the output
  variable is present (i.e., true).

  \item Add to the MAX-SAT model one soft clause with the negation of each other
  variable. The maximization will therefore try to make as few variables present
  as possible.

  \item Add a hard clause for each variable imposing that if it is present, its
  derivative can be represented (with degree at most 2) in the system. This is
  done with an implication: if the variable is true, then take (the CNF
  representation of) the conjunction of all the monomials in its derivative, and
  for each the disjunction of one of its possible representation computed in
  step 3 should be true (i.e., present in the system).

\end{enumerate}
\caption{Encoding of QTP in MAX-SAT.}\label{algo:max-sat}
\end{algorithm}

An example of what happens in step 3 is as follows: assume you get the monomial
\(ab^{2}\) in the derivative of the monomial \(m\). There are three different
ways to represent it: as a single variable \(x_{ab^{2}}\), or as a product
\(x_{a}x_{b^{2}}\) or \(x_{ab}x_{b}\). Hence in step 6 we will get the CNF
representation of
\(x_{m}\Rightarrow(x_{ab^{2}}\vee(x_{a}\wedge x_{b^{2}})\vee(x_{ab}\wedge x_{b}))\vee\ldots\)
More generally, we have

\begin{proposition}
The number of variables in our MAX-SAT model is \(\lvert M\rvert\),
and the number of clauses, because of the DNF-to-CNF conversion is bounded by
\(O(\lvert M\rvert + 2^{d})\), where \(d\) is the highest product of the degrees
of any monomial of \(m'\).
\end{proposition}
\begin{proof}
  Indeed there are less than $d=\frac{1}{2}\prod_{1\leq i\leq n} (d_{i}+1)$ ways to represent, as a product
  (independent of the order) of one or two variables, the monomial $\prod x_{i}^{d_{i}}$.
  This leads, in step 3, to a Boolean representation as a big disjunction of $d$ conjunctions of two variables,
  which once converted to CNF amounts to at worst $2^{d}$ clauses.
\end{proof}

\section{NP-Hardness}\label{sec:np}

In this section and Appendix~\ref{proofNPcomplete} we prove the NP-completeness of nsQTDP, through a reduction of the
Vertex Set Covering Problem (VSCP)~\cite{GJ79book}, i.e.~the problem of determining the minimum number of
vertices that touch every edges of a graph.

We give in Appendix~\ref{Horn} a similar, yet simpler, reduction to show the NP-hardness of the Max-Horn-SAT problem
(while Horn-SAT and Min-Horn-SAT are in P). It may
be useful to the reader to read this proof to help understand the logic of the reduction
before getting into the more complicated details of the differential equation setting.
In essence both reductions work by translating the choice between the two ends of an edge in
a graph in a choice in the other problem. Let us take an edge and its two
vertices that we will call $V_i$ and $V_j$.  For the Horn-SAT problem, we introduce a
clause of the form $\neg v_i \lor \neg v_j$ that ensures that one of the two variables is
set to false in a satisfied instance; setting a variable to false thus indicates that the
corresponding vertex is in the covering.
For the quadratic reduction problem, we introduce a monomial of degree $3$
($V_i V_j Z$) in the derivative of an auxiliary variable. To perform a quadratic
transformation, we then have to ``split'' this monomial as the product of two variables:
$\overline{V_iZ} \times \overline{V_j}$ or $\overline{V_i} \times \overline{V_jZ}$. The
variables of the form $\overline{V_iZ}$ appearing in the reduction will correspond to the
vertex $V_i$ in the covering of the graph.

Another way to see the connection between these reductions lies in the parallel between
Horn-SAT as a model of theorem prover (if B and C are true then A is true) and the
Quadratic Transformation as a model of computation (if variable B and C are computed
monomials then A can be).

\subsection{Encoding of the Vertex Set Covering Problem}\label{encoding}

Given a graph $G = (V,E)$, a vertex cover is a subset of vertices, $S \subset V$, so that every edge
has at least one endpoint in $S$:
\begin{equation}
	\forall e = (i,j) \in E, (i \in S) \lor (j \in S).
\end{equation}

The VSCP is the optimization problem of finding the smallest vertex cover in a given graph:
\begin{description}
\item {\bf instance:} Graph $G=(V, E)$
\item {\bf output:} Smallest number $k$ such that $G$ has a vertex cover of size $k$.
\end{description}
The associated decision problem is to determine the existence of a vertex cover of size at most $k$.

It is well-known that the vertex set covering decision problem is NP-complete.
Here, we prove the same for the non-succinct quadratic transformation problem for PIVPs (nsQTDP).
The general idea is, starting from a graph $G$, to construct a PIVP where only the first
derivatives contains monomials of degree higher than 2, in such a way that the set of 
variables of the output is simply linked with the elements of the optimal cover $S$ of
$G$.

Starting from a graph $G = (\{V_1,\ldots,V_n\},E)$, we construct $\text{PIVP}_3(G)$
with $n+2$ variables, defined by:
\begin{align}
   \frac{dV_0}{dt} &= \sum_{(V_i,V_j) \in E} V_i V_j V_{n+1} + V_1,\\
   \frac{dV_i}{dt} &= \sum^{n+1}_{j=1} a_{i,j} V_i V_j + V_{i+1} \quad \forall i \in [1,n],\\
   \frac{dV_{n+1}}{dt} &= \sum^{n+1}_{j=1} a_{{n+1},j} V_{n+1} V_j,\\
   a_{i,j} &= i(n+2)+j
\end{align}
and an initial condition of the form: $V_i(t=0) = \frac{i}{i+1}$.

It is worth noting that the $a_{i,j}$'s (and the initial conditions) are chosen here just
to be different in each derivative (and variables), this ensures that no polynomial may be
used to quadraticly transformed this PIVP. It is interesting to note that the initial
condition are not essential for the proof and that the quadratic transformation is as hard
for PODE as it is for PIVP.

This encoding shows with a proof given in Appendix~\ref{proofNPcomplete} that

\begin{thm}\label{thm:complete}
The nsQTDP (resp.~nsQTP) is NP-complete (resp.~NP-hard).
\end{thm}

In the succinct representation of the input PIVP by a list of symbolic functions,
if Conjecture~\ref{conjecture} is true,
we get that the witness may have an exponential size in the size of the succinct representation of the input PIVP,
which leads us to:

\begin{conjecture}\label{thm:exp-complete}
The QTDP is NEXP-complete. QTP is NEXP-hard.
\end{conjecture}

\subsection{Minimizing the Number of Monomials}

It is legitimate to ask if minimizing the number of monomials (i.e. reactions in the ECRN
framework) is as hard as minimizing the number of variables (i.e. species). Actually, the proof given
above still works for this variant of QTP:

\begin{thm}\label{thm:reaction-exp-complete}
Given a PIVP $P$ with variables $v_i$, determining a set of variables $v'_j$ defines
through functions $f_j$ of the $v_i$: $v'_j = f_j(\{v_i\})$ such that the PIVP $P'$
thus defined is quadratic, encodes the same function as $P$ and has less than $k$
monomials is an NP-complete problem.
\end{thm}

The proof is given in Appendix~\ref{app:reaction-nphard}.

Now, as shown in the following section,
though of same theoretical complexity as minimizing the number of species,
minimizing the number of reactions seems a bit easier in practice with the MAX-SAT algorithm.

\section{Practical Complexity}\label{sec:evaluation}

\subsection{Benchmark of CRN Design Problems}

The quadratization problem naturally arises in the synthetic biology perspective for the
problem of designing an ECRN to implement a given high-level function presented by a PIVP\@.
We propose here such a benchmark of synthesis problems for sigmoid functions and
particularly Hill functions of various order, and other functions of interest to
understand the practical complexity of QTP\@.

For this article, we were particularly interested in the time taken to find the optimal
solution of the quadratic transformation and as such report the performance for the
resolution of this precise problem.
We therefore provide in
Table~\ref{tab:bench} both the total execution time going from the PIVP to the ECRN (Total
time) and the time taken by the MAX-SAT solver that solves the quadratic transformation
problem while minimizing the number of species (SAT-Sp time).  We also give in the table
the number of variables introduced in Alg.~\ref{algo0}.
along with the optimal number of variables found by our algorithm (Optimal
var.). We finally mention the time taken to minimize the number of reactions (SAT-Reac
time) and the resulting number of reactions (Optimal reac.).
All computation times are given in milliseconds and were obtained on a personal laptop (Lenovo W530,
Intel Core i7-3720QM CPU, 2.60GHz x 8).

\begin{table}
\begin{center}
\begin{tabular}[pos]{|c|c|c|c|c|c|}
  \hline
	CRN & Alg.~\ref{algo:max-sat} & MAX-SAT & Minimum / Alg.~\ref{algo0}  & Min reactions & Min reactions\\
	name & total time & time & nb.~variables  & MAX-SAT time & nb.~reactions \\
        & ms& ms&  &ms&\\
        \hline
	circular 2,3 & 80.35 & 0.2 & 5 / 14 & 0.2 & 6 \\
	circular 2,4 & 120.4 & 0.6 & 6 / 23 & 0.6 & 8 \\
	circular 2,5 & 869.5 & 7.2 & 6 / 34 & 6.6 & 8 \\
	circular 2,6 & 54450 & 754.5 & 7 / 47 & 945.1 & 10 \\
	hard3 & 1576 & 7.3 & 14 / 34 & 7.3 & 28 \\
	hard4 & 28730 & 369.3 & 16 / 43 & 297.5 & 31 \\
	hill2 & 77.74 & 0.1 & 3 / 5 & 0.1 & 3 \\
	hill2x & 90.86 & 0.1 & 5 / 11 & 0.2 & 7 \\
	hill3 & 78.06 & 0.1 & 4 / 8 & 0.1 & 4 \\
	hill3x & 103.5 & 0.2 & 6 / 17 & 0.3 & 9 \\
	hill4 & 85.18 & 0.1 & 5 / 11 & 0.1 & 5 \\
	hill4x & 152.2 & 0.7 & 7 / 23 & 0.7 & 11 \\
	hill5 & 84.7 & 0.2 & 5 / 14 & 0.2 & 5 \\
	hill5x & 543.8 & 5.2 & 7 / 29 & 3.8 & 11 \\
	hill6 & 103.4 & 0.3 & 6 / 17 & 0.3 & 6 \\
	hill6x & 3934 & 60.2 & 8 / 35 & 37.3 & 13 \\
	hill7 & 112.1 & 0.5 & 6 / 20 & 0.4 & 6 \\
	hill7x & 35130 & 1016 & 8 / 41 & 338.7 & 13 \\
	hill8 & 151.1 & 1.3 & 7 / 23 & 1.0 & 7 \\
	hill10 & 580.7 & 10.2 & 7 / 29 & 6.8 & 7 \\
	hill15 & 92850 & 6486 & 8 / 44 & 2908 & 8 \\
	monom 2 & 102.5 & 0.2 & 6 / 7 & 0.3 & 14 \\
	monom 3 & 567.0 & 1.0 & 16 / 25 & 1.9 & 73 \\
	selkov &  87.68 & 0.1 & 4 / 4 & 0.2 & 12 \\
        \hline
\end{tabular}

\caption{Benchmark of quadratization problems given with computation times in ms for the
	tranformation to MAX-SAT and for MAX-SAT solving (Alg.~\ref{algo:max-sat},
  the minimum number of variables compared to the number of variables found by Alg.~\ref{algo0}, and the minimum number of monomials (i.e.~elementary reactions).}\label{tab:bench}\label{table:bench}
\end{center}
\end{table}

Our protocol to gather these results is as follow. We first time the whole process of
compiling the PIVP through the ``compile\_from\_PIVP'' command of Biocham, thus giving the
Total time. During the process we keep the temporary file that were given to the SAT
solver and does a second execution of the SAT solver alone with a verbose output,
gathering the information given by the output to determine the SAT time (doing this twice
for both SAT-Sp and SAT-Reac). Hence, the total time contains the time it takes to
construct and write the cnf file while the MAX-SAT time only measure the resolution of the
formulae by the max sat solver. The time taken to convert the resulting PIVP to the ECRN
language is essentially negligible.

In Table~\ref{tab:bench}, we use the following nomenclature:

``circular$(n,k)$'' denotes a circular PODE with $n$ variables of degree $k$:
\begin{equation}
	\frac{dX_i}{dt} = X_{i+1}^k, \quad \frac{dX_n}{dt} = X_1^k.
\end{equation}
it can be check that introducing all monomials of a single variable ($x, x^2, \ldots$) is
sufficient.

``hard$k$'' models are designed to be especially demanding in terms of monomials, the input
is:
\begin{equation}
	\frac{dA}{dt} = C^k + A^2 B^2 C^{k-1},\quad
	\frac{dB}{dt} = A^2,\quad
	\frac{dC}{dt} = B^2,
\end{equation}
so that while they ask for relatively few variables and are described with a handful of
monomials they actually need most of the variables of the proof making them interesting to
understand the effective structure of the QTP. The construction is based on the one of
circular$(n,k)$ adding a second monomial to the first derivative in order to make
mandatory the usage of variables using several of the old variables.

``monom$n$'' is one of the most promising model regarding the NEXP complexity as it
rely on $n$ variables and a long monomial of size $n$ so that the input is of size $n^2$.
But we suspect it to ask of the order of $2^n$ variables, the input is:
\begin{equation}
	\frac{dX_i}{dt} = X^2_{(i+1)} + \prod^n_{j=1} X^2_j.
\end{equation}
(for clarity we do not add the modulo in the equation but $X_{n+1}$ is the same as $X_1$.)
We were not able to reduce ``long monom $4$'' despite the reduction being very quick on
the $n=3$ case.

``hill$n$'' is the Hill function of order $n$ through the 3 variables PIVP\@:
\begin{equation}
	\frac{dH}{dt} = n I^2 T^{n-1},\quad
	\frac{dI}{dt} = -n I^2 T^{n-1},\quad
	\frac{dT}{dt} = 1.
\end{equation}
so that $H$ is the desired hill function, $I$ is complementary to the hill function
($I+H=1$) and $T$ is an explicit time variable $T=t$.
The ``x'' after the model indicate that the PODE has been modified to take the desired
point of computation as an input, hence the initial concenctration of the $X$ species is
now the input of the computation:
\begin{equation}
	\frac{dH}{dt} = n I^2 T^{n-1} X,\quad
	\frac{dI}{dt} = -n I^2 T^{n-1} X,\quad
	\frac{dT}{dt} = X,
	\frac{dX}{dt} = -X.
\end{equation}

Selkov is a common model of Hopf bifurcation:
\begin{equation}
	\frac{dX}{dt} = -X + a Y + X^2 Y, \quad \frac{dY}{dt} = b - a Y - X^2 Y,
\end{equation}
where $a$ and $b$ are tunable parameters.

\begin{table}
\begin{center}
\begin{tabular}[pos]{|l|l|}
  \hline 
	Model name & Optimal solution with a minimum number of variables \\
	\hline 
	circular 2,3 & $\{ x, y, xy, x^2, y^2 \}$ \\
	circular 2,4 & $\{ x, y, x^2y, xy^2, x^3, y^3 \}$ \\
	circular 2,5 & $\{ x, y, x^3y, xy^3, x^4, y^4 \}$ \\
	circular 2,6 & $\{ x, y, x^4y, x^3y^2, xy^4, x^5, y^5 \}$ \\
	hard3 & $\{ a, b, c, ac, a^2, b^2, a^2b, ab^2, ab^2c, b^2c, c^3, ac^3, b^2c^3, ab^2c^3 \}$ \\
	hard4 & $\{ a,b,c,a^2,b^2,ab^2,a^2b,b^2c,c^3,ac^3,bc^3,a^2c^2,b^2c^2,c^4,bc^4,ab^2c^3 \}$ \\
	hill2 & $\{ i, it, h \}$ \\
	hill2x & $\{ i, h, x, ix, itx \}$ \\
	hill3 & $\{ i, h, it, it^2 \}$ \\
	hill3x &  $\{ i, h, x, ix, itx, it^2x \}$ \\
	hill4 &  $\{ i, h, t, it^2, it^3 \}$ \\
	hill4x & $\{ i, h, x, ix, tx, itx, it^3x \}$ \\
	hill5 &  $\{ i, h, t, t^3, it^4 \}$ \\
	hill5x & $\{ i, h, x, ix, tx, t^3x, it^4x \}$ \\
	hill6 & $\{ i, h, t, it^2, it^4, it^5 \}$ \\
	hill6x & $\{ h, x, ix, tx, it^2x, it^3, it^3x, it^5x \}$ \\
	hill7 & $\{ i, h, t, t^3, it^4, it^6 \}$ \\
	hill7x & $\{ i, h, x, ix, tx, t^3x, it^6x, it^2x \}$ \\
	hill8 & $\{ i, h, t, t^2, t^5, it^6, it^7 \}$ \\
	hill10 & $\{ i, h, t, t^3, t^7, it^8, it^9 \}$ \\
	hill15 & $\{ i, h, t, t^2, t^5, it^{11}, it^{13}, it^{14} \}$ \\
	monom 2 & $\{ a, b, a^2, b^2, a^2b, ab^2 \}$ \\
	monom 3 & $\{
		a,b,c,a^2,b^2,c^2,abc,ab^2,ac^2,a^2b,a^2c,bc^2,b^2c,ab^2c^2,a^2bc^2,a^2b^2c \}$ \\
	selkov & $\{ x,y,xy,x^2 \}$ \\
\hline 
\end{tabular}
\caption{Minimal number of variables and optimal solutions found by Alg.~\ref{algo:max-sat} on our benchmark of QTP instances (Table.~\ref{table:bench}).
}\label{tab:optimal}
\end{center}
\end{table}

\subsection{BioModels Repository}

The BioModels database~\cite{CLN13issb} is a repository of models of natural biological
processes.  Among the $653$ models from the curated branch of BioModels, only $232$ are
reaction models with mass action law kinetics thus leading to polynomial ODEs, among which
only $12$ are of degrees strictly higher than 2.  This is not surprising because the
reaction models in BioModels are mechanistic models naturally described by elementary
CRNs.

The non elementary CRN with mass actions law kinetics of BioModels are models number:
$123$, $152$, $153$, $163$, $281$, $407$, $483$, $530$, $580$, $630$, $635$, $636$.
Currently, our MAX-SAT algorithm fails to solve the QTP optimization problem on those instances in less than one hour computation time.
The encoding in MAX-SAT is itself very long because of exponential size complexity in those cases.

To take an example, the model label $123$ contains $13$ species but only 4 of them
participate in monomials of degree greater than 2, namely degree 4. Manually restricting the quadratic
transformation to this set still gives us a search space of $65$ possible variables, a bit
larger than what is currently handled by our algorithm. Pruning further to select a
smaller subset of the ODE that contains two variables and only one of the two problematic
monomials, gives us a model that is easily solved in a few seconds.
However, that solution does not solve optimally the complete model.

\section{Conclusion}

The problem of CRN design for implementing a given computable real function presented as the solution of a PIVP
has been solved on the theoretical side by the proof of Turing-completeness for finite continuous CRNs \cite{FLBP17cmsb}.
Nevertheless to make that approach practical, good algorithms are needed to eliminate degrees greater than 2 in the PIVP.
Though it is well known in dynamical system theory that there is no loss of generality to consider
polynomial ordinary differential equations with degrees at most 2,
that quadratization problem has apparently not been studied from a computational point of view.

We have shown the NP-hardness of the quadratization optimization problem in the non succinct representation
of the input PIVP by a matrix of monomials,
when we want to minimize either the number of species, or the number of reactions.
In the succinct symbolic representation of the input PIVP by list of monomials,
we conjecture that the problem becomes NEXP-hard.
A proof would need to show that the hard instances coming of the vertex set covering problem used in the proof of NP-completeness,
may have optimal solutions of exponential size in the succinct representation.

Nevertheless, we have shown that an algorithm based an encoding in MAX-SAT
is able to solve interesting CRN design problems in this approach.
A particularly interesting example is the automated synthesis of an abstract CRN of 11 reactions over 7 molecular species to implement the Hill function of order 5
which can be compared to the 10 reactions over 12 species
of the concrete MAPK signalling CRN implementing a similar input/output function \cite{HF96pnas}.

\subsection*{Acknowledgements}
This work was jointly supported by ANR-MOST \emph{BIOPSY Biochemical Programming System} grant ANR-16-CE18-0029
and ANR-DFG \emph{SYMBIONT Symbolic Methods for Biological Networks} grant ANR-17-CE40-0036.

\bibliographystyle{plain}
\bibliography{contraintes,binomial}

\newpage
\section{Appendix: NP-hardness of  MAX-Horn-SAT}\label{Horn}

A Horn clause is a disjunction of literals with at most one positive literal.
Horn-SAT is the problem of deciding the satiafiability of a conjunction of Horn clauses.
Such a problem can be easily solved by unit-clause propagation, as follows
\begin{enumerate}
\item Ignore the clauses that contain both a variable and its negation
\item Set all variables to false
\item Initialize the score of each clause to its number of negative literals
\item For each unsatisfied clause with 0 score
  \begin{enumerate}
\item If it has no positive literal return \emph{Unsatisfiable}
\item Otherwise set the positive literal x to true
\item Decrement the score of the other clauses having x as negative literal
\end{enumerate}
\item Return \emph{Satisfiable}
\end{enumerate}
This algorithm clearly shows that Horn-SAT is in P.
In addition, this algorithm obviously minimizes the number of variables set to true.
Perhaps surprisingly however:

\begin{proposition}
  Deciding the satisfiability of a Horn-SAT instance while asking that at least $k$ variables 
	are set to true is NP-complete and MAX-Horn-SAT is NP-hard.
\end{proposition}

\begin{proof}
This can be easily shown by reduction of the \emph{Vertex Set Covering Problem}.
Given a graph $G$ with $n$ vertices, we introduce one variable $v_i$ for each vertex, and one clause
$\neg v_i \lor \neg v_j$ for each edge $(v_i, v_j)$. A variable set to false indicates
that the corresponding vertex is in the covering.

Now, there is a vertex set covering with $k$ vertices if and only if there is a valuation with $n-k$ variables
set to true satisfying the Horn-SAT instance.

This concludes the proof of NP-completeness and MAX-Horn-SAT is thus NP-hard.
\end{proof}

In essence, the proof of NP-hardness of the non-succinct quadratic transformation follows the same vein but
is quite obfuscated by the details of this problem.

\newpage
\section{Appendix: Proof of NP-completeness of nsQTDP}\label{proofNPcomplete}

In this appendix we prove the NP-completeness of nsQTDP (Thm.~\ref{thm:complete}).
We will construct this proof step by step. In a first time we will describe and study the
encoding of the VSCP as a quadratic reduction, then we will proove that choosing an optimal
set of variables among the ones introduced by the algorithm~\ref{algo0} is an
NP-hard problem. Then we will explain why allowing other types of new variables in the
output (polynomial or algebraic function) still preserved the NP-hardness of the problem.

By abuse of notation, we use the same names for the vertices of $G$ and the variables of
the PIVP. (Except, of course, for $V_0$ and $V_{n+1}$ that do not exist in the initial
graph.) However, to distinguish between the monomials of the various PIVP and the
variable of the output of the algorithm, these variables will be indicated with an upper
bar like: $\overline{V_i V_j}$ while monomials will not. We will moreover say that a
variable is computed while monomial will be
designated as reachable given a certain set of variables.

Let us now investigate the structure of the constructed PIVP.

\begin{lmm}\label{lmm:quadratic_set}
   Supposing that $\{\overline{V_0}, \dots, \overline{V_{n+1}} \}$ are already computed,
   then the derivative of $\overline{V_i V_j}$ is quadratic for the set of variables 
	$\{\overline{V_0}, \dots, \overline{V_{n+1}}, \overline{V_i V_j}\}$.
\end{lmm}
\begin{proof}
Denoting $X = \overline{V_i V_j}$, we have:
\begin{align}
\frac{dX}{dt} &= \frac{dV_i}{dt} V_j + V_i \frac{dV_j}{dt} \\
&= \sum_m (a_{i,m}+a_{j,m}) V_i V_j V_m + V_{i+1} V_j + V_i V_{j+1} \\
&= \sum_m (a_{i,m}+a_{j,m}) X V_m + V_{i+1} V_j + V_i V_{j+1}
\end{align}
where one of the two last term may be missing if $i$ or $j$ is $n+1$. This is
quadratic with respect to the aforementioned set.
\end{proof}

Hence, if all the initial variables are present, we can add or remove variables of
degree two, knowing that there derivatives will always be quadratic. In effect, this allows
us to focus on the monomials in the derivative of $V_0$ as the only monomials that will
need new variables to be reachable.

This property does not hold for variable of degree $3$ (and all higher degree). Indeed, in that case,
the last monomials are of degree $3$ (and higher) and so need a way to be computed either
by introducing them entirely as new variables or relying on breaking them between variables
of lesser degree that may not be already computed. The derivative of the variable of
degree 3 $\overline{V_i V_j V_k}$ present for example the non-quadratic monomials:
$V_{i+1} V_j V_k$, $V_i V_{j+1} V_k$, $V_i V_j V_{k+1}$.

This is also false for polynomial variables because the derivative all have different
rates $a_{i,j}$ that ensure that a polynomial do not appear in its own derivative as a
monomial does. Thus computing a polynomial variable may ask us to compute still other variables.

This property is essential for our proof as it allows us to make a direct connexion between vertex covering and quadratic
reduction, namely that
given a cover $S = S_1,\ldots,S_k$ of $G$, the set of functions:
\begin{equation}
	\{ \overline{V_0},\ldots,\overline{V_{n+1}},\overline{S_1 V_{n+1}},\ldots,\overline{S_k
	V_{n+1}} \} \label{eq:optimal}
\end{equation}
have $n+2+k$ elements and defines a quadratic transformation of $PIVP_3(G)$.

It is indeed obvious that the derivative of $\overline{V_0}$ is quadratic using the fact that every
edge in $E$ has at least one endpoint in $S$ and so each triplet may be rewritten with two
of the new variables. Checking that the other variables also have quadratic derivatives is
easy given Lemma~\ref{lmm:quadratic_set}.

To prove that our transformation is valid however, we need the opposite! We want to check
that an optimal transformation of $\text{PIVP}_3(G)$ effectively allows us to find an
optimal vertex covering of $G$. And essentialy, we will do this by showing that optimal
reduction are of the form of the set~\ref{eq:optimal} thus making a direct connexion
between optimal covering and quadratic reduction.

Essentialy, the remainder of the proof will be to demonstrate the following lemma:
\begin{lmm}\label{lm:np_proof}
	For a given graph $G$, optimal reductions of $PIVP_3(G)$ may be rewritten in the form
	of equation~\ref{eq:optimal} and thus define an optimal vertex cover of $G$.
\end{lmm}

\subsection{Restriction of Variables to Monomials Functions}

As explained above, we first prove that finding an optimal set of variables among the
monomials described in the paper of Carothers (see Alg.~\ref{algo0}) is an NP-hard problem.
This give us the soften version of Lemma~\ref{lm:np_proof}:

\begin{lmm}
	Given a graph $G$, the smallest subset of variables considered in
	Alg.~\ref{algo0} that forms a quadratic transformation of
	$PIVP_3(G)$ gives an optimal vertex cover of $G$.
\end{lmm}

\begin{proof}

As expressed above, we want to show that optimal reductions are of the form of
	Equation~\ref{eq:optimal}, or at least may be easily reshape to be so.

By definition, we need to introduce the first variable $\overline{V_0}$, the derivative of
	which present the term $\overline{V_1}$, thus asking us to compute it too. Then in
	turn, it asks us to compute $\overline{V_2}$ and so on until all the variables of
	degree one are present.

Let us take an optimal quadratic transformation, then the different monomials in the
	derivative of $V_0$ are reachable. This means that if we have a monomial like $V_i V_j V_{n+1}$,
	at least one of the four following variables is present: $\overline{V_i V_{n+1}}$,
	$\overline{V_j V_{n+1}}$, $\overline{V_i V_j}$ or $\overline{V_i V_j V_{n+1}}$. If we
	are in the third or fourth case, we can remove this variable and replace it by
	$\overline{V_i V_{n+1}}$. (As all variables of degree one are present, we know that
	this new variable preserves the quadraticity of the solution and is thus still optimal.)
	Moreover, by the structure of $PIVP_3(G)$ variables like these appear only once in
	the derivative of $V_0$, thus this transformation still allows us to compute the
	desired function.

	As we then have all variables of degree one and that all the other variables are of
	degree two, we know that the PIVP is quadratic by Lemma~\ref{lmm:quadratic_set}, moreover
	we cannot have increase the number of variables and are thus still optimal. Finally,
	we have construct a set like Eq.\ref{eq:optimal} and have thus defined an optimal
	covering $S$ of $G$ with the optimal transformation of $\text{PIVP}_3(G)$.
\end{proof}

To generalize the previous proof to any set of monomials, we note that
by construction, the set of Sec.~\ref{Carothers} overspan the set of variables that may be used to
define a quadratic transformation. In particular due to the construction of $PIVP_3$, it
contains all monomials of degree one and two that can be formed with the variables of the
input PIVP, and all the monomials that appears in the derivatives of $PIVP_3$.

Hence, a monomial that is not present in the mathematical proof can only increase the number of
monomials that need to be reached as it appears nowhere and will need to be computed
itself. It thus cannot be present in an optimal set.

\subsection{Restriction of Variables to Polynomials Functions}
\begin{lmm}
   No polynomial variable is present in an optimal quadratic transformation of $PIVP_3$.
\end{lmm}
\begin{proof}
The idea is still the same. We want to prove that we need to introduce all the variables
of degree one and once this is done, that using only the variables that correspond to the vertex
cover is preferable. But it is now more tricky as the ending singulets of the derivative
may be added to a polynomial to ``save a variable''.

For the same reason as before, $\overline{V_0}$ needs to be computed. To investigate why the other
variables $\overline{V_i}$ are also needed is more complex.

Suppose we wish to avoid computing the variable $V_k$ so that we add it in an existing
	polynomial (eventually composed of a single monomial)
hence forming the variable $M = P + V_k$, where $P$ is some polynomial. Let us look at
its derivative:
$\frac{dM}{dt} = \frac{dP}{dt} + \frac{dV_k}{dt}$

As noted above, and due to the presence of the parasitic terms $a_{i,j}$, $M$ do not
	appear in its own derivative. Thus, a transformation like the one of
	Lemma~\ref{lmm:quadratic_set} is out of hope. Moreover, the derivative of $V_k$ present
	a term in $V_k^2$. To compute it we can either add $\overline{V_k}$ to our set of
	variable which is what we try to avoid, either add $\overline{V_k^2}$ (or a polynomial
	incorporating it). But you can check that the derivatives of such a variable present
	a term in $V_k^3$. So either we abdicate and include the variable of degree one, either
	you add a polynomial of degree 3, but this polynomial will ask us a new one of degree
	4, etc. To avoid an infinite set of variables we have to compute $\overline{V_k}$, and
	this is true for all $k$. Thus all variables of the initial PIVP need to be present.

Now, for each monomials in the derivatives of the first variable, we have $2$ choices on
the way it is computed. Either a single variable is introduced to deal with it and this has
already been treated in the previous case. Either all or part of it is computed as part of
a polynomial. To prove that this cannot be done in an optimal transformation we need to show
that doing so imply to compute additional undesired variables. And once again we can
convince ourself by inspecting such derivatives, for example:

\begin{align}
   \frac{d}{dt}(P + V_i V_j) &= \frac{dP}{dt} + \frac{dV_i V_j}{dt} \\
    &= \frac{dP}{dt} + \sum_m (a_{i,m} + a_{j,m}) V_i V_j V_m +\ldots 
\end{align}

cannot be quadratic if the variable $\overline{V_i V_j}$ is not computed, which we try to
   avoid  or another more complex polynomial specificaly tailored for this purpose. Thus,
   trying to hide a part of a monomial in a polynomial to save a variable always ask at
   least two variables and cannot be part of an optimal transformation.

\end{proof}

\subsection{Quadratic Transformation Without Restriction}

Finally, we notice that for a function to be in the output, it have to be polynomial as
it will actually be used to rewrite polynomial functions. Hence, putting the previous
results together we get:
\begin{proposition}\label{VSCP}
	A graph $G$ with $n$ vertices has a vertex set cover of size $k$ if and only if
	PIVP$_3$(G) has a quadratic transformation to a PIVP of dimension $n+k+2$.
\end{proposition}
\noindent
which with Prop.~\ref{inNP} concludes the proof of Thm.~\ref{thm:complete}.

\subsection{Proof of NP-Hardness for Reactions minimization}\label{app:reaction-nphard}

Thm.~\ref{thm:reaction-exp-complete}
\begin{proof}
The core of the proof is similar, using the same reduction from VSCP. Starting
	from a graph $G$ with $n$ vertices and $\ell$ edges, we construct $PIVP_3(G)$.
	
	As in the previous case, we still have to introduce all the variables of the form
	$\overline{V_i}$ giving us a fixed number of monomials upon which no optimization is
	possible. Let us note $F(n,\ell) = n^2 + 3n + \ell + 2$ this number.
	
	Then, introducing a variable like $\overline{V_i V_j}$ imposes $n+3$ monomials if $i,j
	\neq n+1$ and $n+2$ if $i,j=n+1$.  As we have seen in the previous proof, the optimal
	cover set may be expressed using only variables like $\overline{V_i}$ and
	$\overline{V_i V_{n+1}}$, and will thus ask for $k = F(n,\ell) + k_s (n+2)$ where $k_s$
	is the number of vertices in the optimal covering of $G$. The main difference with the
	proof for variables is that we do not have to check variables of the form
	$\overline{V_i V_j}$ as they ask one more monomial than the one with $j = n+1$ and
	are thus never optimal.
\end{proof}

\end{document}